\documentclass{article}
\usepackage{arxiv}

\usepackage{xcolor}  

\usepackage{amssymb}

\usepackage{amsmath}
\usepackage{mathdots}

\usepackage{mathtools}

\usepackage{tensor}

\usepackage{urwchancal}
\DeclareFontFamily{OT1}{pzc}{}
\DeclareFontShape{OT1}{pzc}{m}{it}{<-> s * [1.10] pzcmi7t}{}
\DeclareMathAlphabet{\mathpzc}{OT1}{pzc}{m}{it}

\usepackage{graphicx}
\usepackage{ifpdf}
\ifpdf
\usepackage{epstopdf}
\PrependGraphicsExtensions{.svg}
\fi

\newtheorem{theorem}{Theorem}[section]
\newtheorem{lemma}[theorem]{Lemma}
\newtheorem{corollary}[theorem]{Corollary}

\newtheorem{remark}[theorem]{Remark}

\providecommand{\R}{\mathbb{R}}

\providecommand{\SO}{\mathbf{SO}}
\providecommand{\SL}{\mathbf{SL}}
\providecommand{\GL}{\mathbf{GL}}
\providecommand{\SE}{\mathbf{SE}}

\providecommand{\grpG}{\mathbf{G}}

\providecommand{\gothgl}{\mathfrak{gl}}

\providecommand{\gothg}{\mathfrak{g}}

\providecommand{\Sph}{\mathrm{S}}

\providecommand{\calJ}{\mathcal{J}}

\providecommand{\calM}{\mathcal{M}}

\providecommand{\calT}{\mathcal{T}}

\providecommand{\cset}[2]{\left\{ {#1} \;\middle\vert\; {#2} \right\}}

\providecommand{\Sym}{\mathbb{S}} 

\providecommand{\tT}{\mathrm{T}} 

\DeclareMathOperator{\tr}{tr}

\DeclareMathOperator{\Ad}{Ad}

\DeclareMathOperator{\image}{im}

\providecommand{\td}{\mathrm{d}}
\providecommand{\tD}{\mathrm{D}}

\providecommand{\ddt}{\frac{\td}{\td t}}

\providecommand{\Hess}{\mathrm{Hess}}

\providecommand{\mr}[1]{\mathring{#1}} 

\providecommand{\ob}[1]{\overline{#1}} 
\usepackage{accents}
\usepackage{mathtools}
\makeatletter
\providecommand{\scirc}{%
    \hbox{\fontfamily{\rmdefault}\fontsize{0.4\dimexpr(\f@size pt)}{0}\selectfont{\raisebox{-0.52ex}[0ex][-0.52ex]{$\circ$}}}}

\makeatother

\makeatletter
\providecommand{\ucirc}{%
    \hbox{\fontfamily{\rmdefault}\fontsize{0.4\dimexpr(\f@size pt)}{0}\selectfont{\raisebox{0.0ex}[0ex][-0.52ex]{$\circ$}}}}

\makeatother

\mathchardef\mhyphen="2D

\providecommand{\etal}{\textit{et al.~}}

\newcommand{\QED}{$\blacksquare$}
\newenvironment{proof}{\noindent\hspace{2em}{\itshape Proof: }}{\hspace*{\fill}~\QED\par\endtrivlist\unskip}
\newenvironment{proofx}[1]{\noindent\hspace{2em}{\itshape Proof of #1: }}{\hspace*{\fill}~\QED\par\endtrivlist\unskip}
\usepackage{hyperref}

\newcommand{\at}[2]{\left. #1 \right|_{#2}}
\newcommand{\bbH}{\mathbb{H}}
\newcommand{\JGradient}{\eta}
\newcommand{\actionMat}{\Upsilon}
\newcommand{\actionMatT}{\ob{\Upsilon}}

\newcommand{\change}[1]{{\color{blue} #1}}
\newcommand{\robchange}[1]{{\color{teal} #1}}
\renewcommand{\change}[1]{#1}
\renewcommand{\robchange}[1]{#1}

\newcommand{\papercitation}{
van Goor, P., Mahony, R. (2024), ``Global Minimum Energy State Estimation for Embedded Nonlinear Systems with Symmetry'', \emph{Accepted for presentation at IEEE CDC}.
}

\newcommand{\publicationdetails}
{
\papercitation \\
\copyrightNoticeIEEE{2024} 
}
\newcommand{\publicationversion}
{Author accepted version}

\title{Global Minimum Energy State Estimation for Embedded Nonlinear Systems with Symmetry}
\headertitle{Global Minimum Energy State Estimation for Embedded Nonlinear Systems with Symmetry}

\author{
    \href{https://orcid.org/0000-0003-4391-7014}{\includegraphics[scale=0.06]{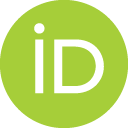}\hspace{1mm}
Pieter van Goor}
\\
    Robotics and Mechatronics Group \\
	University of Twente \\
    7522 NB Enschede, Netherlands \\
    \texttt{p.c.h.vangoor@utwente.nl} \\
	\And	\href{https://orcid.org/0000-0002-7803-2868}{\includegraphics[scale=0.06]{orcid.png}\hspace{1mm}
    Robert Mahony}
\\
    Systems Theory and Robotics Group \\
	Australian National University \\
    ACT, 2601, Australia \\
	\texttt{robert.mahony@anu.edu.au} \\
}

\begin{document}

\maketitle
\thispagestyle{empty}
\pagestyle{empty}

\begin{abstract}
Choosing a nonlinear state estimator for an application often involves a trade-off between local optimality (such as provided by an extended Kalman filter) and (almost-/semi-) global asymptotic stability (such as provided by a constructive observer design based on Lyapunov principles).
This paper proposes a filter design methodology that is both global and optimal for a class of nonlinear systems.
In particular, systems for which there is an embedding of the state-manifold into Euclidean space for which the measurement function is linear in the embedding space and for which there is a synchronous error construction.
A novel observer is derived using the minimum energy filter design paradigm and exploiting the embedding coordinates to solve for the \emph{globally optimal solution exactly}.
The observer is demonstrated through an application to the problem of unit quaternion attitude estimation, by embedding the 3-dimensional nonlinear system into a 4-dimensional Euclidean space.
Simulation results demonstrate that the state estimate remains optimal for all time and converges even with a very large initial error.
\end{abstract}

\section{Introduction}

State estimation for nonlinear systems is a challenging problem and has many applications across robotics, aerospace, computer vision, etc.
The extended Kalman filter (EKF) is the de facto standard solution, and while it and its relatives (UKF~\cite{2004_julier_UnscentedFilteringNonlinear}, IEKF~\cite{2018_barrau_InvariantKalmanFiltering}, EqF~\cite{2023_vangoor_EquivariantFilterEqF}, etc.) can provide a locally optimal state estimate, no guarantees can be made about their performance beyond a local domain \cite{1999_reif_StochasticStabilityDiscretetime}.
In contrast, constructive nonlinear observer designs that exploit Lie theory and differential geometry have demonstrated global and almost-global asymptotic stability, but do not provide the same local optimality as an EKF-type approach
\cite{2008_mahony_NonlinearComplementaryFilters,2007_baldwin_ComplementaryFilterDesign,2010_lageman_GradientLikeObserversInvariant,2012_mahony_NonlinearComplementaryFilters,2021_berkane_NonlinearNavigationObserver,2018_zlotnik_GradientbasedObserverSimultaneous}.
\robchange{
A key building block in constructive nonlinear observer design is the concept of synchrony: the property that the observer-system error dynamics are linearly dependent on the observer correction term \cite{2010_lageman_GradientLikeObserversInvariant}; that is, there are no exogenous or drift terms in the error dynamics.
This property makes designing the correction term for a Lyapunov observer design straightforward \cite{2023_vangoor_SynchronousObserverDesign}.}
Synchrony is also closely related \cite{2023_vangoor_SynchronousObserverDesign} to the group-affine structure that plays a fundamental role in the performance of the Invariant EKF \cite{2018_barrau_InvariantKalmanFiltering}.
The stochastic motivation for an EKF (or IEKF) design, however, does not generalise well to global nonlinear systems analysis.

Minimum energy filtering provides a best-of-both-worlds deterministic view of the noisy state estimation problem \cite{1968_mortensen_MaximumlikelihoodRecursiveNonlinear}
that specialises to the (optimal) Kalman filter on linear systems \cite{1968_mortensen_MaximumlikelihoodRecursiveNonlinear} but can be formulated globally on Lie-groups and homogeneous spaces \cite{2016_saccon_SecondOrderOptimalMinimumEnergyFilters}.
Although the minimum energy framework has been used to derive EKF type filters
\cite{2011_zamani_MinimumenergyFilteringUnit,2011_zamani_NearOptimalDeterministicFiltering,2013_zamani_MinimumEnergyFilteringAttitude,2016_saccon_SecondOrderOptimalMinimumEnergyFilters,2023_rigo_SecondorderoptimalFilteringSE}
by taking second order approximations of the value function evolution, to the authors' knowledge, no optimal exact solution has been found prior to this work.
\change{There is also a significant body of work in geometric optimal control on Lie-groups and homogeneous spaces \cite{1997_jurdjevic_GeometricControlTheory,2015_bloch_NonholonomicMechanicsControl,2005_bullo_GeometricControlMechanical} that considers the optimal control problem, but to the authors understanding, does not treat the filtering problem.}

In this paper, we consider the state estimation problem for a class of systems posed on homogeneous spaces that can be algebraically embedded into an ambient Euclidean space.
In addition, we require that the systems have dynamics and measurement processes that can be written as linear in the embedding coordinates.
Note that while the dynamics and measurement processes may be written algebraically in a linear form, the system state is constrained to lie on a nonlinear embedded manifold and the actual system is therefore not linear in a classical sense.
The approach has been considered in the past:
Choukroun \etal \cite{2006_choukroun_NovelQuaternionKalman} proposed to estimate the quaternion attitude of a vehicle by treating the system as a linear system in $\R^4$ and then renormalising the state estimate to recover a unit quaternion.
Aguiar and Hespanha \cite{2006_aguiar_MinimumenergyStateEstimation} provided a rigid-body pose estimation algorithm that relied on an embedding of the orientation matrix into $\R^{3\times 3} \simeq \R^9$, and required a normalisation step to recover a valid rotation matrix.
In both of these examples, optimality of the solution was lost in the projection step that was needed to ensure the state estimate remained on the state-manifold.
In contrast, in this paper we exploit Lie-group symmetry to ensure satisfaction of the state constraint while exploiting linearity of the embedding system structure to obtain an exact optimal solution.

Our contributions are as follows:
\begin{itemize}
    \item We provide an exact optimal observer solution to a \emph{nonlinear} minimum energy filter problem for systems on homogeneous spaces that can be expressed as algebraically linear in embedding coordinates.
    \item We demonstrate our method in an application to optimal quaternion attitude estimation.
    To the authors understanding, this is the first globally optimal (in the sense of a particular minimum energy cost criteria) attitude filter.
\end{itemize}

The key advantage of the proposed state estimator is that it is both optimal and global.
EKF approaches are optimal but only in a local domain, while deterministic approaches behave well globally but cannot provide (even local) optimality in general.

\section{Preliminaries}

The inner product and norm on $\R^{m\times n}$ are defined by
\begin{align*}
    \langle A, B \rangle &:= \tr(A^\top B), &
    \vert A \vert &:= \sqrt{\langle A, A \rangle},
\end{align*}
for all $A,B \in \R^{m\times n}$.
The set of symmetric positive-definite $m\times m$ matrices is denoted $\Sym_+(m)$.
For any $P \in \Sym_+(m)$, the weighted vector norm is defined by
\begin{align*}
    \vert v \vert_P &:= \sqrt{\langle P v, v \rangle},
\end{align*}
for every $v \in \R^m$.

For a differentiable function $V : \R^m \to \R^n$, we denote the differential of $V$ at $x$ by $\tD V(x) \in \R^{n\times m}$ with
\begin{align*}
    \tD V(x) u &:= \at{\ddt}{t=0} V(x + tu).
\end{align*}
For functions of multiple variables, a subscript may be attached to $\tD$ to indicate the variable being differentiated.
Square brackets are used, $\tD V(x) [u]$, when the matrix structure of $\tD V(x)$ multiplying $u \in \R^m$ is unclear.
If $V$ is scalar-valued, i.e. $V : \R^m \to \R$, then the gradient is defined to be
\begin{align*}
    \nabla V(x) = \tD V(x)^\top \in \R^m,
\end{align*}
since the differential $\tD V(x)$ is a $1\times m$ row vector.
In this case, the Hessian of $V$ at $x$ is defined to be
\begin{align*}
    \Hess V(x) = \tD \nabla V(x) \in \R^{m\times m}.
\end{align*}

Let $\grpG \leq \GL(m)$
\footnote{We write $\grpG \leq \GL(m)$ rather than $\grpG \subseteq \GL(m)$ to emphasise that $\grpG$ is not only a subset but also a subgroup.}
be an $m\times m$ matrix Lie-group with manifold dimension $d = \dim \grpG$.
We denote the Lie-algebra $\gothg \leq \gothgl(m)$
\footnote{Similarly to the group case, the use of $\leq$ rather than $\subseteq$ emphasises that $\gothg$ is a Lie subalgebra (not just a subset) of $\gothgl(m)$.}
, and define the `wedge' operator $\cdot^\wedge : \R^d \to \gothg$ and `vee' operator $\cdot^\vee : \gothg \to \R^d$ to be an identification of the vector space underlying $\gothg$ with $\R^d$; that is, for each $u \in \R^d$, there is a unique corresponding vector $u^\wedge \in \gothg$, and $(u^\wedge)^\vee = u$.
Note that the wedge and vee operators can always be defined, although they are not unique and depend on a choice of basis for $\gothg$.
We define the matrices $\actionMat_\xi \in \R^{m \times d}$ and $\actionMatT_\xi \in \R^{m \times d}$ by
\begin{align}\label{eq:action_matrices}
    \actionMat_\xi \, \Delta^\vee &:= \Delta \xi, &
    \actionMatT_\xi \, \Delta^\vee &:= \Delta^\top \xi,
\end{align}
for every $\xi \in \R^m$ and $\Delta \in \gothg \leq \gothgl(m)$.

\section{Problem Description}

\change{
We consider a class of systems on homogeneous spaces that admit an expression as algebraically linear in embedding coordinates.
This includes systems on matrix Lie-groups where the measurement is a linear group action, such as attitude estimation on $\SO(3)$ from bearing measurements \cite{2008_mahony_NonlinearComplementaryFilters} and rigid body pose estimation on $\SE(2)$ or $\SE(3)$ from landmark measurements \cite{2006_aguiar_MinimumenergyStateEstimation,2007_baldwin_ComplementaryFilterDesign}.
This also includes systems on homogeneous spaces such as quaternion attitude estimation on $\Sph^3$ \cite{2006_choukroun_NovelQuaternionKalman} (also see \S~\ref{sec:example}) and bearing estimation on $\Sph^2$ \cite{2008_mahony_NonlinearComplementaryFilters}.
However, this excludes systems where the output is given by a nonlinear group action, such as homography estimation on $\SL(3)$ from image features \cite{2012_mahony_NonlinearComplementaryFilters}.
}

Let $\grpG \leq \GL(m)$ be a $d$-dimensional matrix Lie-group with Lie-algebra $\gothg \leq \gothgl(m)$, and consider the embedded manifold $\calM \subseteq \R^m$ defined by
\begin{align}
    \calM = \grpG \mr{\xi} = \cset{X^{-1} \mr{\xi} \in \R^m}{ X \in \grpG }, \label{eq:manifold_dfn}
\end{align}
for some fixed \emph{origin} $\mr{\xi} \in \R^m$.
\change{Note that $\calM$ is a homogeneous space of $\grpG$ by definition.}
Then, for any $\xi \in \calM$, the tangent space at $\xi$ is given by
\begin{align}
    \tT_\xi \calM := \cset{U \xi \in \R^m}{ U \in \gothg} = \cset{\actionMat_{\xi} U^\vee }{ U \in \gothg}.
    \label{eq:tangent_space_dfn}
\end{align}
Consider a deterministic nonlinear system model defined on $\calM$ with outputs in $\R^n$ of the form
\begin{subequations}\label{eq:system_dfn}
    \begin{align}
        \dot{\xi}(t) &= - U_t \xi(t),  \label{eq:system_dfn_dynamics}\\
        y(t) &= C_t \xi(t)  \label{eq:system_dfn_measurement}
    \end{align}
\end{subequations}
where $U_t \in \gothg$ and $C_t \in \R^{n\times m}$ are known time-varying matrices.
This is an \emph{ideal system} that does not take into account the imperfections and noise inherent in real-world systems.
\change{Note that we use the brackets around $t$ in $\xi(t)$ to indicate this is a solution to an ODE or a part of the design process, while we use the subscript $t$ in $U_t$ to indicate that this is a known parameter.
In the sequel we will use $(t)$ when introducing variables but may then drop this explicit notation for the sake of readability.}

The dynamics \eqref{eq:system_dfn_dynamics} lift to classical left invariant dynamics on the Lie-group
\[
\dot{X}(t) = X(t) U_t
\]
via the group action \eqref{eq:manifold_dfn};
that is, $\ddt \xi(t) = \ddt (X^{-1} \mr{\xi}) = - U_t \xi$.
This is a common model encountered in real-world systems with linear group actions such as directions under rotation, positions of landmarks under rigid-body translation, etc.
Note that the output \eqref{eq:system_dfn_measurement} is linear with respect to the embedding of the manifold in $\R^m$.
This is a key assumption in the present work that is also common in examples such as \cite{2005_mahony_ComplementaryFilterDesign,2007_baldwin_ComplementaryFilterDesign,2006_choukroun_NovelQuaternionKalman}.


\begin{remark}
    While it may appear straightforward to simply treat \eqref{eq:error_system} as a linear system, this approach will not take into account the manifold constraint $\xi \in \calM$.
    The problem with this is twofold.
    First, a general linear-systems filter for \eqref{eq:system_dfn} will not enforce the constraint and there is no intrinsic method to reproject a general $\hat{\xi} \in \R^m$ onto the manifold $\calM$ \cite{2006_aguiar_MinimumenergyStateEstimation,2006_choukroun_NovelQuaternionKalman}.
    Second, the system may be observable when restricted to $\calM$ without being observable when considered over the ambient space $\R^m$.
\end{remark}

\subsection{Error System}

We introduce an observer or reference state in order to rewrite the system dynamics in error form.
Define an \emph{observer state} $\hat{X} \in \grpG$ with dynamics
\begin{align}\label{eq:reference_trajectory}
    \dot{\hat{X}}(t) &= \hat{X}(t) U_t + \Delta(t) \hat{X}(t), & \hat{X}(0) = \hat{X}_0
\end{align}
where $\Delta \in \gothg$ is a correction term that remains to be designed.
For now, $\Delta$ will be treated as an arbitrary time-varying signal, and the error system will be analysed without making any assumptions about how $\Delta$ is chosen.
Given a trajectory $\xi(t) \in \calM$ of the system \eqref{eq:system_dfn}, the \emph{equivariant error} \cite{2023_vangoor_EquivariantFilterEqF} is
\begin{align*}
    e(t) &:= \hat{X}(t) \xi(t) \in \calM \hookrightarrow \R^m.
\end{align*}
The ideal system \eqref{eq:system_dfn} and observer dynamics \eqref{eq:reference_trajectory} then combine to yield an \emph{ideal error system},
\begin{subequations}\label{eq:error_system_nominal}
\begin{align}
    \dot{e}(t) &= \Delta(t) e(t), \label{eq:error_system_nominal_dyn}\\
    y(t) &= C_t \hat{X}(t)^{-1} e(t).\label{eq:error_system_nominal_meas}
\end{align}
\end{subequations}

To capture the fact that real-world systems may not obey their models exactly, we modify the ideal error system model to include unknown error signals.
Define the \emph{error system}
\begin{subequations}\label{eq:error_system}
    \begin{align}
        \dot{e}(t) &= \Delta(t) e(t) +  B_t \mu(t), \label{eq:error_system_dyn}\\
        y(t) &= C_t \hat{X}(t)^{-1} e(t) + \nu(t), \label{eq:error_system_meas}
    \end{align}
\end{subequations}
where:
\begin{itemize}
    \item $B_t \in \R^{m \times l}$ is a chosen input matrix.
    \item $\mu(t) \in \R^l$ is the unknown dynamics error.
    \item $\nu(t) \in \R^n$ is the unknown output error.
\end{itemize}
Note that these dynamics are determined by the choice of $\Delta \in \gothg$.
The role of the input matrix $B_t$ is to constrain the directions in which the dynamics error $\mu(t)$ can affect $\dot{e}(t)$.
In practice, $B_t$ is associated with the geometric constraints of the state space (see Section \ref{sec:gain_selection}), but the developments in this paper will treat it simply as a known, arbitrary, time-varying matrix.
For fixed signals $U_t, C_t, y(t), B_t, \hat{X}(t), \Delta(t)$, we define the set
\begin{align}\label{eq:trajectory_set}
    \calT_{[0,t]} = \cset{(e_{[0,t]},\mu_{[0,t]},\nu_{[0,t]})}{\text{$e,\mu,\nu$ satisfy \eqref{eq:error_system}}}.
\end{align}
The notation $e_{[0,t]}$ emphasises that we refer to a trajectory in $\R^m$ rather than a single point $e \in \R^m$.
\change{
For a trajectory $(e_{[0,t]},\mu_{[0,t]},\nu_{[0,t]}) \in \calT_{[0,t]}$, the initial condition $e(0)$ and the signal $\mu_{[0,t]}$ determine $e{[0,t]}$ by the error dynamics \eqref{eq:error_system_dyn}, and this in turn determines $\nu_{[0,t]}$ by the error system measurement \eqref{eq:error_system_meas}.
}
The error signals $\mu(t)$ and $\nu(t)$ are modelled as deterministic noise signals; that is, \emph{a-priori} unknown time-sequences that do not necessarily have a stochastic interpretation \cite{1968_mortensen_MaximumlikelihoodRecursiveNonlinear,2013_zamani_MinimumEnergyFilteringAttitude}.
The problem addressed in this paper is to find an error state $e^\star(t) \in \calM$ that is optimal in terms of a cost functional that penalises large error signals $\mu(t)$ and $\nu(t)$, subject to the dynamics and measurements \eqref{eq:system_dfn}.

\subsection{Minimum Energy Filtering}

Classical minimum energy filtering chooses the correction term to minimise a cost functional, defined in terms of the error signals, subject to the system dynamics \eqref{eq:system_dfn}.
In this paper, we study the error system \eqref{eq:error_system} around an arbitrary known observer trajectory $\hat{X}$ \eqref{eq:reference_trajectory} rather than working directly with the original system definition.
Provided an optimal error trajectory $e^\star$, a corresponding state estimate is recovered by
\[
\hat{\xi} = \hat{X}^{-1} e^\star(t).
\]

The optimality criterion in this paper is developed in terms of the error system \eqref{eq:error_system}.
Define the \emph{initial cost} $V_0 : \R^m \to \R^+$ and the \emph{stage cost} $\ell_\tau : \R^l \times \R^n \to \R^+$ to be
\begin{align}
    V_0(e(0)) &:= \frac{1}{2} \vert e(0) - \hat{X}(0) \xi_0 \vert_{H_0}^2, \label{eq:initial_cost} \\
    \ell_\tau(\mu(\tau), \nu(\tau)) &:=  \frac{1}{2} \vert \mu(\tau) \vert^2_{Q_\tau}
    + \frac{1}{2} \vert \nu(\tau) \vert^2_{R_\tau}, \label{eq:state_cost}
\end{align}
in terms of the state error $\mu(t)$ and output error $\nu(t)$,
where
\begin{itemize}
    \item $\xi_0 \in \calM$ is the \emph{initial state estimate}.
    \item $H_0 \in \Sym_+(m)$ is the \emph{initial gain}.
    \item $Q_\tau \in \Sym_+(l)$ is the \emph{state gain}.
    \item $R_\tau \in \Sym_+(n)$ is the \emph{output gain}.
\end{itemize}
The \emph{cost functional} $\calJ_t : \calT_{[0,t]}  \to \R^+$ is defined to be
\begin{align*}
    \calJ_t(e_{[0,t]},\mu_{[0,t]},\nu_{[0,t]}) &:= V_0(e(0)) + \int_0^t \ell_\tau(\mu(\tau), \nu(\tau)) \td \tau.
\end{align*}
The cost functional is a measure of how large the error terms $\mu_{[0,t]}$ and $\nu_{[0,t]}$ must be chosen in order that the error trajectory $e_{[0,t]}$ satisfies \eqref{eq:error_system}.
\change{
In other words, the cost functional penalises large deviations $\mu$ and $\nu$ of the error system \eqref{eq:error_system} from the nominal error system \eqref{eq:error_system_nominal}.
}

The \emph{value function} $V_t : \R^{m} \to \R^+$ is defined as the cost of the final state $e(t)$ given that the rest of the trajectory is chosen to minimise $\calJ$.
That is,
\begin{align}\label{eq:value_function_dfn}
    V_t(e) := \min\cset{\calJ(e_{[0,t]}, \mu_{[0,t]}, \nu_{[0,t]})}{e_{[0,t]}(t) = e}
\end{align}
It is straightforward to verify that, as the notation suggests, the initial condition $V_0$ of $V_t$ is exactly the initial cost defined in \eqref{eq:initial_cost}.
The Hamilton-Jacobi-Bellman equation \cite{2005_bertsekas_DynamicProgrammingOptimal} is applied in the following Lemma and corollaries to provide recursive formulae for the evolution of $V_t$, its gradient $\nabla_e V_t$, and its Hessian $\Hess_e V_t$.
The proofs are provided in the Appendix.

\begin{lemma}\label{lem:value_function_derivative}
The time-differential of the value function $V_t$ defined in \eqref{eq:value_function_dfn} is given by
\begin{align}
    \tD_t V_t(e) &= - \langle \nabla_e V_t(e), \Delta e \rangle
    - \frac{1}{2} \vert B_t^\top \nabla_e V_t(e) \vert_{Q_t^{-1}}^2
    \notag \\ &\hspace{1cm}
    + \frac{1}{2} \vert y - C_t \hat{X}^{-1} e \vert^2_{R_t}.
\end{align}
\end{lemma}

\begin{corollary}\label{cor:value_function_gradient_derivative}
    The gradient of the value function \eqref{eq:value_function_dfn} evolves according to
        \begin{align}
        \tD_t \nabla_e V_t(e)
        &=- \Hess_e V_t(e)[\Delta e]
        - \Delta^\top \nabla_e V_t(e)
        \notag \\ &\hspace{1cm}
        - \Hess_e V_t(e) B_t Q_t^{-1}B_t^\top \nabla_e V_t(e)
        \notag \\ &\hspace{1cm}
        + \hat{X}^{-\top} C_t^\top R_t(C_t \hat{X}^{-1} e - y).
        \label{eq:gradient_derivative}
        \end{align}
\end{corollary}

\begin{corollary}\label{cor:hess}
The Hessian $H(t) = \Hess_e V_t(e)$ of the value function \eqref{eq:value_function_dfn} is independent of $e$ and has dynamics
\begin{align}
    \dot{H} &= - H \Delta
    - \Delta^\top H
    - H B_t Q_t^{-1}B_t^\top H
    \notag \\ &\hspace{1cm}
    + \hat{X}^{-\top} C_t^\top R_t C_t \hat{X}^{-1}.\label{eq:hessian_derivative}
\end{align}
\end{corollary}


\section{Observer Design}

The insight of Mortensen \cite{1968_mortensen_MaximumlikelihoodRecursiveNonlinear} was that the optimal solution of the minimum energy filter is characterised by tracking the critical point of $V_t(e)$.
In the constrained case, where $e \in \calM$ is enforced, a \emph{constrained critical point} is a critical point of the value function \emph{restricted to the manifold}  $V_t|_\calM : \calM \to \R$.
In other words, a constrained critical point is a value $e \in \calM$ for which $\langle \nabla V_t(e), \zeta \rangle \equiv 0$ for all $\zeta \in \tT_{e(t)} \calM$.
This means that the gradient of the value function $V_t$ only needs to be nullified in the tangent space of the manifold at $e$, and may be non-zero in the normal space (the orthogonal complement of the tangent space).

Choose $\mr{\xi} := \xi_0$ and choose $\hat{X}_0 \in \grpG$ such that $\hat{X}_0 \xi_0 = \mr{\xi}$, noting that $\hat{X}_0 \in \grpG$ may not be fully determined by this condition but any such choice is sufficient.
It is easily verified that the critical point of $\nabla V_0(e(0))$ at time $t = 0$ is $e^\star(0) = \mr{\xi} \in \calM$.
The approach taken here is to design the observer $\hat{X}$ trajectory by defining $\Delta \in \gothg$ to ensure $e(t) = e^\star = \mr{\xi}$ is a constrained critical point for all time.
Note that this means that $\Delta$ will depend on the measurements $y(t)$ and $U_t$, and on the input, output, and gain matrices.
This is compatible with the analysis undertaken so far, since no assumptions were made on $\Delta$ other than that it is not a function of the trajectory optimisation parameters $(e,\mu,\nu)_{[0,t]}$.

\begin{theorem}
\label{thm:critical_point}
Recall the observer dynamics \eqref{eq:reference_trajectory}.
Define $\Delta \in \gothg$ to satisfy
\begin{align}\label{eq:observer_definition}
    P &:= \actionMat_{\mr{\xi}}^\top H \actionMat_{\mr{\xi}}
    + \actionMat_{\mr{\xi}}^\top \actionMatT_{\JGradient}, \\
    P \Delta^\vee &= \actionMat_{\mr{\xi}}^\top \left(
        \hat{X}^{-\top} C_t^\top R_t(C_t \hat{X}^{-1} \mr{\xi} - y)
        - H B_t Q_t^{-1} B_t^\top \JGradient
        \right), \notag
\end{align}
where $H \in \Sym_+(m)$ denotes the Hessian of $V_t$, and $\JGradient = \nabla_e V_t(\mr{\xi}) \in \R^m$ denotes the gradient of $V_t$ at $\mr{\xi}$.
Let $V_t \vert_\calM$ denote the restriction of the value function $V_t$ \eqref{eq:value_function_dfn} to the manifold $\calM$.
Then $e^\star = \mr{\xi} \in \calM$ is a critical point of $V_t \vert_\calM$ for all time.
\end{theorem}

\begin{proof}
\change{
The condition that $e^\star = \mr{\xi} \in \calM$ is a critical point of $V_t$ can be expressed as $\nabla_e V_t \vert_\calM (\mr{\xi}) = 0$, which is equivalent to
\begin{align}\label{eq:optimality_condition}
    \actionMat_{\mr{\xi}}^\top \nabla_e V_t(\mr{\xi}) = 0,
\end{align}
Since $\tT_{\mr{\xi}} \calM = \image \actionMat_{\mr{\xi}}$.
It is immediately clear that this is satisfied at the initial condition $t=0$.
For any $t \geq 0$, differentiating \eqref{eq:optimality_condition} yields
\begin{align}\label{eq:differential_optimality_condition}
    \ddt \actionMat_{\mr{\xi}}^\top \nabla_e V_t(\mr{\xi}) = \actionMat_{\mr{\xi}}^\top \tD_t \nabla_e V_t(\mr{\xi}) = 0.
\end{align}
The remainder of the proof will show that this condition is indeed satisfied for all time.
}

Using the shorthand $\eta = \nabla_e V_t(\mr{\xi})$ and substituting in the dynamics \eqref{eq:gradient_derivative}, the optimality condition becomes
\begin{align*}
    0 &= \actionMat_{\mr{\xi}}^\top \dot{\JGradient} \\
    &= \actionMat_{\mr{\xi}}^\top \left(
        - H \Delta \mr{\xi} - \Delta^\top \JGradient
    - H B_t Q_t^{-1}B_t^\top \JGradient
    \right. \\ &\hspace{1.5cm}\left.
    + \hat{X}^{-\top} C_t^\top R_t(C_t \hat{X}^{-1} \mr{\xi} - y)
    \right) \\
    &= \actionMat_{\mr{\xi}}^\top \left(
        - H \actionMat_{\mr{\xi}} - \actionMatT_J \right) \Delta^\vee +
    \actionMat_{\mr{\xi}}^\top \left(- H B_t Q_t^{-1}B_t^\top \JGradient
    \right. \\ &\hspace{1.5cm}\left.
    + \hat{X}^{-\top} C_t^\top R_t(C_t \hat{X}^{-1} \mr{\xi} - y)
    \right) \\
    &= - P \Delta^\vee +
    \actionMat_{\mr{\xi}}^\top \left(- H B_t Q_t^{-1}B_t^\top \JGradient
    \right. \\ &\hspace{2.5cm}\left.
    + \hat{X}^{-\top} C_t^\top R_t(C_t \hat{X}^{-1} \mr{\xi} - y)
    \right),
\end{align*}
where $P$ is defined as in \eqref{eq:observer_definition}.
Thus, the differential optimality condition \eqref{eq:differential_optimality_condition} is precisely satisfied when $\Delta$ is defined by the proposed equation \eqref{eq:observer_definition}.
It follows that $\mr{\xi}$ remains a critical point of $V_t$ for all time.
\end{proof}

Minimising the constrained value function $V_t|_\calM$ is equivalent to minimising the cost functional $\calJ_t$ over trajectories $e_{[0,t]}$ that end on the manifold, $e_{[0,t]}(t) \in \calM$.
By choosing $\Delta$ as described in Theorem \ref{thm:critical_point}, the origin $\mr{\xi} = e^\star$ is always the minimum energy solution of the constrained value function.
In other words, $\mr{\xi}$ is always the minimum energy `estimate' for the current error state $e^\star_{[0,t]}(t) \in \calM$.
Although the derivation of the filter is done in terms of the error state, the equivalent state estimate can be recovered easily by
\begin{align}
    \hat{\xi}(t) &:= \hat{X}(t)^{-1} e^\star_{[0,t]}(t) = \hat{X}(t)^{-1} \mr{\xi} \in \calM.
\end{align}
Crucially, the state estimate $\hat{\xi}$ will remain in the manifold for all time due to the relationship \eqref{eq:manifold_dfn}.

\subsection{Selection of Gain Matrices}
\label{sec:gain_selection}

The gain matrices $H_0, Q_t, R_t$ and the input matrix $B_t$ can be chosen to reflect the noise characteristics of a real-world system.
The initial Hessian $H_0$ represents the confidence of the initial state estimate information $\xi_0 \in \R^m$, and typically, it is sensible to choose $H_0$ small.
The state gain matrix $Q_t$ reflects confidence in the error dynamics, and $Q_t^{-1}$ represents the corresponding uncertainty.
In a practical system, much of the uncertainty in the error dynamics is associated with uncertainty in the measured velocity.
Let $\hat{U}_t \in \gothg$ denote a measurement of the velocity signal $U_t \in \gothg$, and assume that $\hat{U}_t^\vee$ is normally distributed with mean $U_t^\vee$ and a given covariance $M_t \in \Sym_+(d)$.
Then if $\xi \in \calM$ denotes the system state, one has
\begin{align*}
    \dot{e} &= \ddt \hat{X} \xi
    = \hat{X} (\hat{U}_t - U_t) \xi + \Delta \hat{X} \xi \\
    &= \Delta e + \hat{X} (\hat{U}_t - U_t) \hat{X}^{-1} e \\
    &= \Delta e + \actionMat_e \Ad_{\hat{X}}^\vee (\hat{U}_t - U_t)^\vee \\
    &\approx \Delta e + \actionMat_{\mr{\xi}} \Ad_{\hat{X}}^\vee (\hat{U}_t - U_t)^\vee,
\end{align*}
where the final step uses the approximation $e \approx \mr{\xi}$, and $\Ad_{\hat{X}}^\vee : \R^d \to \R^d$ is the linear map defined by $\Ad_{\hat{X}}^\vee U^\vee = (\hat{X} U \hat{X}^{-1})^\vee$.
We thus choose to define $B_t := \actionMat_{\mr{\xi}}\Ad_{\hat{X}}^\vee$ and $Q_t = M_t^{-1}$.

The output gain matrix $R_t$ represents the confidence in the measurement $y \in \R^n$, and $R_t^{-1}$ represents the uncertainty.
Let $\hat{z} = h(\xi) + \nu' \in \R^{n'}$ denote a physical measurement of the system for some output map $h :\calM \to \R^n$, where $\nu' \in \R^{n'}$ is a normally distributed noise term with zero-mean and covariance $\bar{R}^{-1} \in \Sym_+(n')$.
If $h(\xi) = C_t \xi$ for a matrix $C_t$, then we can simply choose $R = \bar{R}$.

In some cases (see the example in Section \ref{sec:example}), it is necessary to write an implicit measurement function to obtain the embedded form \eqref{eq:system_dfn_measurement}; that is,
\begin{align*}
    y = C_t(h(\xi)) \xi,
\end{align*}
where $C_t = C_t(h(\xi))$ is a matrix built from the physical measurement $h(\xi)$.
In this situation, we may employ linearisation to obtain a model for the error.
Let $z = h(\xi)$, then
\begin{align*}
    y(z, e) &= C_t (z) \hat{X}^{-1} e \\
    &= C_t(\hat{z} - \nu') \hat{X}^{-1} e, \\
    &\approx C_t(\hat{z}) \hat{X}^{-1} e - \tD C_t (\hat{z})[\nu']  \hat{X}^{-1} e \\
    &\approx C_t(\hat{z}) \hat{X}^{-1} e - \tD C_t (\hat{z})[\nu'] \hat{X}^{-1} \mr{\xi} \\
    &= C_t(\hat{z}) \hat{X}^{-1} e + \tD_z y(\hat{z}, \mr{\xi})[z - \hat{z}]
\end{align*}
where the first approximation linearises the implicit measurement $y(z, e) \in R^{n}$ in terms of $z$, and the second approximation once more uses $e \approx \mr{\xi}$.
In this case we thus choose to define
\begin{align*}
    R_t^{-1} = \tD_z y(\hat{z}, \mr{\xi}) \bar{R}^{-1} \tD_z y(\hat{z}, \mr{\xi})^\top.
\end{align*}

\section{Example: Quaternion Attitude Estimation}
\label{sec:example}
Quaternions provide a convenient setting to study the problem of attitude estimation provided angular velocity measurements and body-frame measurements of a time-varying reference vector.

\subsection{Problem Description of the Example}

Define $\bbH$ to be the set of unit quaternions, and define $q \in \bbH$ to be the unit quaternion describing a vehicle's orientation with respect to some inertial reference frame.
Note that $-q$ describes the same physical attitude as $q$.
We identify the space of quaternions $\bbH$ with the 3-sphere $\Sph^3 \subset \R^4$ by writing $q \in \bbH$ as a vector $(q_r, q_v) \in \R^4$, where $q_r \in \R$ is the real component of $q$ and $q_v \in \R^3$ is the vector of imaginary components of $q$.
The product of two quaternions $q, h \in \bbH$ is given by
\begin{align*}
    (q_r, q_v) * (h_r, h_v)
    &= (q_r h_r - q_v^\top h_v, \; q_r h_v + h_r q_v + q_v \times h_v).
\end{align*}

The dynamics of $q = (q_r, q_v)$ are given by
\begin{align}
    \dot{q} &= q * (0, \; \frac{1}{2}\omega); \text{ \change{i.e.,}}\notag\\
    \ddt \begin{pmatrix}
        q_r \\ q_v
    \end{pmatrix}
    &= -\frac{1}{2}\begin{pmatrix}
        0 & \omega^\top \\
        -\omega & \omega^\times
    \end{pmatrix}
    \begin{pmatrix}
        q_r \\ q_v
    \end{pmatrix},
    \label{eq:quaternion_matrix_dynamics}
\end{align}
where $\omega \in \R^3$ is a measurement of the angular velocity of the vehicle expressed in the body-frame, $*:\bbH \times \bbH \to \bbH$ denotes the quaternion product, and $\omega^\times \in \R^{3\times 3}$ is defined by the property that $\omega^\times v = \omega \times v$ for all $v \in \R^3$.

As shown by the matrix form \eqref{eq:quaternion_matrix_dynamics}, the dynamics of $q$ are linear with respect to the embedding of $\bbH \hookrightarrow \R^4$, and match the form \eqref{eq:system_dfn_dynamics}, where the Lie-algebra $\gothg$ and wedge operator are defined by
\begin{align*}
    \gothg &:= \cset{\delta^\wedge \in \R^{4\times 4}}{ \delta \in \R^3} \leq \gothgl(4), &
    \delta^\wedge &:= \begin{pmatrix}
        0 & \delta^\top \\
        -\delta & \delta^\times
    \end{pmatrix}.
\end{align*}
The vee operator $\cdot^\vee : \gothg \to \R^3$ is simply defined as the inverse of the wedge.
The corresponding Lie-group $\grpG \leq \GL(4)$ is defined by
\begin{align*}
    \grpG := \cset{\cos(\theta) I_4 + \sin(\theta) \delta^\wedge}
    {\delta \in \R^3,\; \vert \delta \vert = 1, \; \theta \in \Sph^1}
\end{align*}
which may be obtained by exponentiating elements of the Lie-algebra.
From the definition of $\gothg$, the matrices \eqref{eq:action_matrices} are
\begin{align*}
    \actionMat_q &=
    \begin{pmatrix}
        - q_v^\top \\
        q_r I_3 + q_v^\times
    \end{pmatrix} \in \R^{4\times 3}, &
    \actionMatT_q = - \actionMat_q  \in \R^{4\times 3}.
\end{align*}

Let $\mr{z} \in \R^3$ be a known time-varying reference vector in the inertial frame.
If $z \in \R^3$ is a measurement of $\mr{z}$ in the body-frame of the vehicle, then
\begin{align}
    \label{eq:quaternion_real_measurements}
    (0, z) = q^{-1} * (0, \mr{z}) * q.
\end{align}
This is a nonlinear measurement, but its defining equation can be manipulated to yield an implicit measurement that is linear with respect to the embedding $\bbH \subset \R^4$ (cf. Section \ref{sec:gain_selection}).
Specifically, by left-multiplying both sides by $q$ and subtracting, one has
\begin{subequations}
\begin{align}
    q * (0, z) - (0, \mr{z}) * q &= 0; \text{ \change{i.e.,}} \\
    \begin{pmatrix}
        0 & \mr{z}^\top - y^\top \\
        z - \mr{z} & -(z + \mr{z})^\times
    \end{pmatrix} \begin{pmatrix}
        q_r \\ q_v
    \end{pmatrix} &= 0.
    \label{eq:quaternion_matrix_measurement}
\end{align}
\end{subequations}
In other words, by treating $z$ as if it were an independent time-varying quantity, a linearly-embedded measurement function can be obtained, defined by
\begin{align*}
    y(z, q) &= C_t(z) q = 0, \\
    C_t &:= C_t(z) = \begin{pmatrix}
        0 & \mr{z}^\top - z^\top \\
        z - \mr{z} & -(z + \mr{z})^\times
    \end{pmatrix}.
\end{align*}

\begin{remark}
A similar interpretation of the dynamics and measurements was proposed in \cite{2006_choukroun_NovelQuaternionKalman}.
However, the solution proposed here is substantially different to this prior work.
While \cite{2006_choukroun_NovelQuaternionKalman} proposed to estimate the quaternion as a free vector in $\R^4$, our approach estimates the unit quaternion in $\bbH$ directly.
That is, in contrast to prior work, our proposed method does not require any form of normalisation to be added to the filter; the estimated quaternion will always remain on the manifold of unit-length quaternions due to the constraints inherent in updating the observer state $\hat{X} \in \grpG$.
\end{remark}

\subsection{Simulation Results}

The dynamics \eqref{eq:quaternion_matrix_dynamics} and measurement \eqref{eq:quaternion_matrix_measurement} are both of the required form \eqref{eq:system_dfn}.
As such, we define the observer $\hat{X} \in \grpG \leq \GL(4)$ to have dynamics
\begin{align*}
    \dot{\hat{X}} = -\frac{1}{2}\hat{X} \omega^\wedge + \Delta \hat{X},
\end{align*}
where $\Delta \in \gothg$ is determined according to \eqref{eq:observer_definition}.

To verify the proposed observer design, we carried out simulations of the dynamics \eqref{eq:quaternion_matrix_dynamics} with measurements \eqref{eq:quaternion_real_measurements}.
The initial condition was chosen to be $q(0) = (1, 0_3) \in \bbH$, and the input signal was determined by
\begin{align*}
    U_t &= \frac{1}{2} \omega(t)^\wedge, &
    \omega(t) &= \begin{pmatrix}
        0.1\cos(0.1t) & 0.0 & 0.2
    \end{pmatrix}^\top \in \R^3.
\end{align*}
The time-varying reference vector $\mr{z}$ was determined by
\begin{align}
    \mr{z}(t) = \begin{pmatrix}
        \sin(t) & 0 & \cos(t)
    \end{pmatrix}^\top \in \R^3,
    \label{eq:reference_vector_dfn}
\end{align}
and the measurement $z(t)$ was computed according to \eqref{eq:quaternion_real_measurements}.
The system was simulated using Lie-Group Euler integration at 0.1~s increments over a total time of 100~s.

The observer was simulated in two settings: in a noise-free setting, and in a setting where the measured input signal $\omega$ and measured output signal $z$ were corrupted with noise.
For the simulations with noise, the measured values $\hat{\omega}$ and $\hat{z}$ were obtained by
\begin{align*}
    \hat{\omega} &= \omega + \bar{\mu}, &
    \bar{\mu} &\sim N(0, \bar{Q}^{-1}) &
    \bar{Q}^{-1} &= 0.01^2 I_3, \\
    \hat{z} &= z + \bar{\nu}, &
    \bar{\nu} &\sim N(0, \bar{R}^{-1}) &
    \bar{R}^{-1} &= 1.0^2 I_3.
\end{align*}
The initial conditions of the observer were chosen to be
\begin{align*}
    \mr{q} &= (1, 0_3), \\
    \hat{X}(0) &= \exp\left(\frac{\pi}{2} (0.99,\; 0, \;0)^\wedge\right), \\
    H_0 &= (0.1)^2 \; \hat{X}(0)^\top \actionMat_{\hat{q}(0)}^\top \actionMat_{\hat{q}(0)} \hat{X}(0).
\end{align*}
In particular, these choices mean that the initial estimated attitude $\hat{q}(0) = \hat{X}_0^{-1} \mr{q}$ was misaligned with the true attitude $q$ by $0.99\pi$~rad.
The gain matrices were chosen according to the procedure described in Section \ref{sec:gain_selection}; specifically,
\begin{align*}
    Q_t^{-1} &= (0.5)^2 \bar{Q}_t^{-1}, \\
    B_t &= \actionMat_{\mr{\xi}} \Ad_{\hat{X}}^\vee, \\
    R_t^{-1} &= \tD_z y(\hat{z}, \hat{q}) \bar{R}^{-1} \tD_z y(\hat{z}, \hat{q})^\top, \\
    \tD_z y(\hat{z}, \hat{q}) &= \begin{pmatrix}
        - \hat{q}_v^\top \\ \hat{q}_r I_3 + \hat{q}_v^\times
    \end{pmatrix}.
\end{align*}
The observer equations \eqref{eq:observer_definition} were implemented with adaptive Lie-group Euler integration.
The integration time-step $\td t$ was chosen such that $\vert \Delta^\vee \td t \vert \leq 0.01$ and $\td t \leq 0.1$.

\begin{figure}
    \centering
    \includegraphics[width=0.7\linewidth]{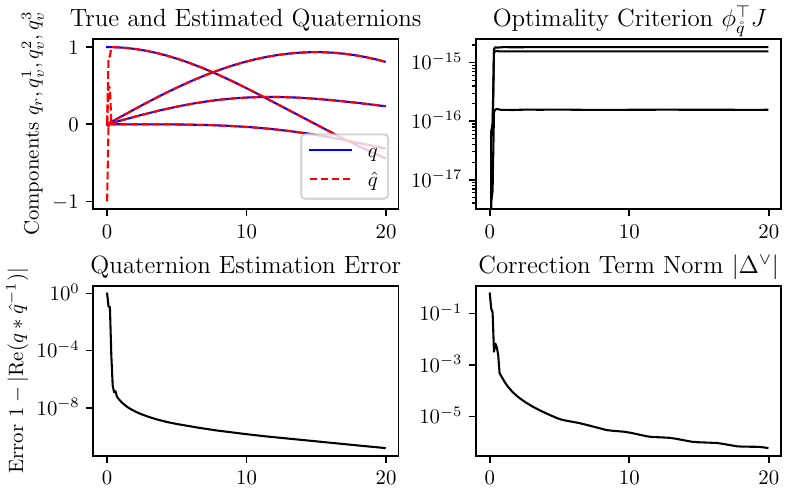}
    \caption{Simulated quaternion attitude estimation without noise.}
    \label{fig:results_noiseless}
\end{figure}

Figure \ref{fig:results_noiseless} shows the results of the estimation without noise.
The initial convergence of the estimate is very fast thanks to the optimality of the proposed observer, in spite of the large initial error of $0.99\pi$~rad.
Following this initial period, the estimation error and correction term both converge exponentially to zero, with oscillations due to the periodic time-variation of the reference vector \eqref{eq:reference_vector_dfn} used in the measurements.
The top-right subplot shows the component values of the optimality criterion \eqref{eq:optimality_condition}.
These all remain near zero, and only fail to be exactly equal to zero due to numerical errors.

\begin{figure}
    \centering
    \includegraphics[width=0.7\linewidth]{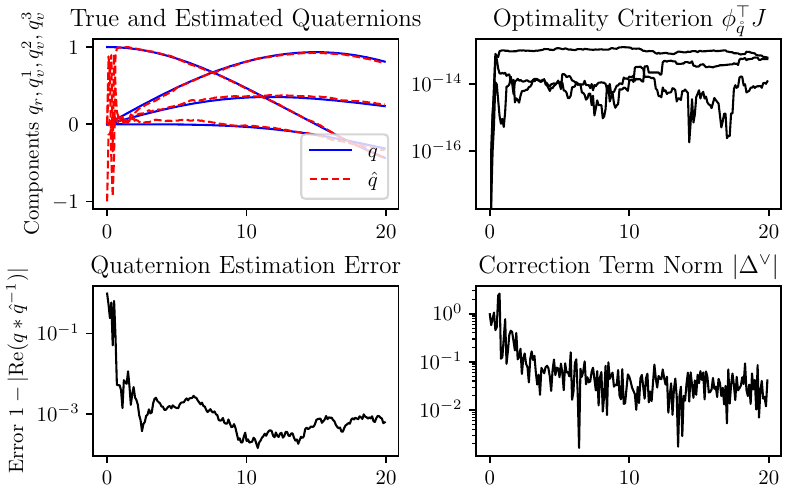}
    \caption{Simulated quaternion attitude estimation with noise.}
    \label{fig:results_noise}
\end{figure}

Figure \ref{fig:results_noise} shows the results of the estimation subject to noise.
As in the noiseless case, the observer is able to estimate the true attitude quickly even with a large initial error.
The optimality criterion is also still satisfied up to numerical errors.
This is as expected, since the observer \eqref{eq:observer_definition} makes no assumption that the input and output signals are free from noise, rather, it only seeks to track a critical point of the value function by keeping the optimality criterion equal to zero.

\section{Conclusion}

This paper proposes a novel observer design for a class of embedded nonlinear systems with Lie-group symmetry.
By leveraging both the embedding structure and the symmetry of the system, the observer is able to provide an estimate of the state that is optimal, globally valid, and always remains on the manifold.
This is a significant contrast to existing approaches, which make trade-offs between local optimality and global validity, or fail to preserve the nonlinear manifold structure intrinsic to the problem.
Simulation results verify the theoretical results and demonstrate performance in a realistic attitude estimation problem.


\bibliographystyle{plain}
\bibliography{2024_minimum_energy.bib}


\appendix
\section{Proofs}

\begin{proofx}{Lemma \ref{lem:value_function_derivative}}
    Applying the HJB equation \cite{2005_bertsekas_DynamicProgrammingOptimal} to \eqref{eq:value_function_dfn} yields
    \begin{align*}
        &\tD_t V_t(e)
        = \min_{\mu \in \R^l} \left\{-\tD_e V_t(e)[\dot{e}] + \ell_t(\mu, \nu) \right\} \\
        &= \min_{\mu \in \R^l} \left\{-\tD_e V_t(e)[\Delta e + B_t \mu] + \ell_t(\mu, \nu) \right\} \\
        &= \min_{\mu \in \R^l} \left\{-\tD_e V_t(e)[\Delta e + B_t \mu] + \frac{1}{2} \vert \mu \vert^2_{Q_t}
        + \frac{1}{2} \vert \nu \vert^2_{R_t} \right\}.
    \end{align*}
    The argument of this minimisation problem is positively quadratic in terms of $\mu$, so the unique minimiser $\mu^*$ is found by setting the derivative with respect to $\mu$ to zero.
    If $\mu^*$ is the minimiser, then for any direction $\delta \in \R^l$,
    \begin{align*}
        0&= -\tD_e V_t(e)[B_t \delta] + \langle Q_t \mu^*, \delta \rangle \\
        &= -\langle B_t^\top \nabla_e V_t(e), \delta \rangle + \langle Q_t \mu^*, \delta \rangle \\
        &= \langle Q_t \mu^* - B_t^\top \nabla_e V_t(e), \delta \rangle,
    \end{align*}
    and hence
    \begin{gather*}
        \mu^* =  Q_t^{-1} B_t^\top \nabla_e V_t(e).
    \end{gather*}
    Substituting this into the HJB equation yields
    \begin{align*}
        \tD_t V_t(e)
        &= -\tD_e V_t(e)[\Delta e + B_t Q_t^{-1} B_t^\top  \nabla_e V_t(e)]
        \\&\hspace{0.5cm}
        + \frac{1}{2} \vert Q_t^{-1} B_t^\top \nabla_e V_t(e) \vert^2_{Q_t}
        + \frac{1}{2} \vert \nu \vert^2_{R_t},
        \\
        &=  - \langle \nabla_e V_t(e), \Delta e \rangle
        - \frac{1}{2} \vert B_t^\top \nabla_e V_t(e) \vert_{Q_t^{-1}}^2
        \notag \\ &\hspace{1cm}
        + \frac{1}{2} \vert y - C_t \hat{X}^{-1} e \vert^2_{R_t},
    \end{align*}
    as required.
\end{proofx}

\begin{proofx}{Corollary \ref{cor:value_function_gradient_derivative}}
    The result is a consequence of Lemma \ref{lem:value_function_derivative} obtained by differentiating with respect to $e \in \R^{m}$.
    For any direction $\delta \in \R^m$, one has
    \begin{align*}
        &\tD_t (\tD_e V_t(e)[\delta])
        = \tD_e (\tD_t V_t(e)) [\delta ], \\
        &= \tD_e \left(
            - \langle \nabla_e V_t(e), \Delta e \rangle
            - \frac{1}{2} \vert B_t^\top \nabla_e V_t(e) \vert_{Q_t^{-1}}^2
            \right.\notag \\ &\hspace{2cm}\left.
            + \frac{1}{2} \vert y - C_t \hat{X}^{-1} e \vert^2_{R_t}
         \right)[\delta], \\
        &= - \langle \nabla_e V_t(e), \Delta \delta \rangle
            - \langle \Hess_e V_t(e) \delta, \Delta e \rangle
            \notag \\ &\hspace{0.5cm}
            - \langle B_t^\top \nabla_e V_t(e), Q_t^{-1} B_t^\top \Hess_e V_t(e)\delta \rangle
            \notag \\ &\hspace{0.5cm}
            + \langle R_t (C_t \hat{X}^{-1} e - y), C_t \hat{X}^{-1} \delta \rangle, \\
        &= - \langle \Delta^\top \nabla_e V_t(e), \delta \rangle
        - \langle \Hess_e V_t(e) \Delta e, \delta \rangle
        \notag \\ &\hspace{0.5cm}
        - \langle \Hess_e V_t(e) B_t Q_t^{-1} B_t^\top \nabla_e V_t(e), \delta \rangle
        \notag \\ &\hspace{0.5cm}
        + \langle \hat{X}^{-\top} C_t^\top R_t (C_t \hat{X}^{-1} e - y), \delta \rangle.
    \end{align*}
    The result follows from the definition of the gradient as the transpose of the differential.
\end{proofx}

\begin{proofx}{Corollary \ref{cor:hess}}
This result follows from differentiating the gradient presented in Corollary \ref{cor:value_function_gradient_derivative}.
The time-derivative of the Hessian is given by
\begin{align*}
    &\tD_t \Hess_e V_t(e) = \tD_e \tD_t \nabla_e V_t(e) \\
    &= - (\tD_e \Hess_e V_t(e))[\cdot, \Delta e]
    - \Hess_e V_t(e) \Delta
    - \Delta^\top \Hess_e V_t(e)
    \notag \\ &\hspace{1cm}
    - (\tD_e \Hess_e V_t(e))[\cdot, B_t Q_t^{-1}B_t^\top \nabla_e V_t(e)]
    \notag \\ &\hspace{1cm}
    - \Hess_e V_t(e) B_t Q_t^{-1}B_t^\top \Hess_e V_t(e)
    \notag \\ &\hspace{1cm}
    + \hat{X}^{-\top} C_t^\top R_t C_t \hat{X}^{-1}.
\end{align*}
Substitute $H$ for $\Hess_e V_t(e)$, and note that $\tD_e H = 0$.
Then we recover exactly the proposed dynamics \eqref{eq:hessian_derivative} for $H$.
It follows from the Cauchy-Kowalevski Theorem \cite{1995_folland_IntroductionPartialDifferential}, $H = \Hess_e V_t(e)$ for all time.
\end{proofx}

\end{document}